\documentclass{llncs}

\usepackage{amsmath}
\usepackage{amssymb}
\usepackage{amsfonts}
\usepackage{color}
\usepackage{xspace}

\newcommand{\ignore}[1]{}
\renewcommand{\Pr}{{\bf Pr}}
\newcommand{\E}{{\bf E}}

\newcommand{\OO}{{\cal O}}

\begin{document}
\title{Optimal Randomized Group Testing Algorithm\\  to Determine the
  Number of Defectives}
\author{
{\bf Nader H. Bshouty}  \\
{\bf  Catherine A. Haddad-Zaknoon} \\ Dept. of Computer Science\\ Technion,  Haifa\\ \ \\
{\bf Raghd Boulos} \\
{\bf Foad Moalem}  \\
{\bf Jalal Nada} \\
{\bf Elias Noufi} \\
{\bf Yara Zaknoon} \\ The Arab Orthodox College, Haifa\\
}
\institute{}
\maketitle

\begin{abstract}
We study the problem of determining exactly the number of defective items in an adaptive Group testing by using a minimum number of tests. We improve the existing algorithm and prove a lower bound that shows that the number of tests in our algorithm is optimal up to small additive terms.
\end{abstract}
\section{Introduction}

Let $X$ be a set of {\it items} with some {\it defective items} $I\subseteq X$.
In Group testing, we {\it test} a subset $Q\subseteq X$ of items, and the answer
to the test is $1$ if $Q$ contains at least one defective item, i.e., $Q\cap I\not=\O$,
and $0$ otherwise. Group testing was initially introduced as a potential approach to the economical mass blood testing, \cite{D43}. However it has been proven to be applicable in a variety of problems, including DNA library screening, \cite{ND00}, quality control in product testing, \cite{SG59}, searching files in storage systems, \cite{KS64}, sequential screening of experimental variables, \cite{L62}, efficient contention resolution algorithms for multiple-access communication, \cite{KS64,W85}, data compression, \cite{HL02}, and computation in the data stream model, \cite{CM05}. See a brief history and other applications in~\cite{Ci13,DH00,DH06,H72,MP04,ND00} and references therein.

Estimating or determining exactly the number of defective items is an important problem in biological and medical applications~\cite{ChenS90,Swallow85}. It is used to estimate the proportion of organisms capable of transmitting the aster-yellows virus in a natural population of leafhoppers~\cite{Thompson62}, estimating the infection rate of the yellow-fever virus in a mosquito population~\cite{WalterHB80} and estimating the prevalence of a rare disease using grouped samples to preserve individual anonymity~\cite{GatwirthH89}.

In the {\it adaptive algorithm}, the tests can depend on the answers to the previous ones. In the {\it non-adaptive algorithm}, they are independent of the previous one and; therefore, one can do all the tests in one parallel step.

In this paper, we study the problem of determining exactly the number of defective items with adaptive Group testing algorithms. We first give an algorithm that improves the number of tests in the best-known algorithm by a factor of~$4$. Improving constant factors in Group testing algorithms is one of the utmost important challenges in group testing since, in many applications, tests are incredibly costly and time-consuming~\cite{BDKS16,CD08,CDL09,DH06,DHWZ06,PR11}. We also give a lower bound that shows that our algorithm is optimal up to a small additive term. To the best of our knowledge, this is the first non-trivial lower bound for this problem.

\subsection{Previous and New Results}
Let $X$ be a set of $n$ items with $d$ defective items $I$. All the algorithms in this paper are {\it adaptive}. That is, the tests can depend on the answers to the previous ones. All the non-adaptive algorithms must ask at least $\Omega((d^2/\log d)\log n)$ queries for determining exactly the number of defective items and $\Omega(\log n/\log\log n)$ for estimating their number, \cite{Bshouty18,DamaschkeM10,DamaschkeM10b}.  In~\cite{BshoutyBHHKS18}, Bshouty et al. show that any deterministic or Las Vegas adaptive algorithm must ask at least $\Omega(d\log (n/d))$.
Since the query complexity depends on the number of items $n$, which, for most applications, is extremely large, non-adaptive algorithms and Las Vegas (and deterministic) algorithms are not desirable for solving this problem.

In \cite{Cheng11}, Cheng gave a randomized Monte Carlo adaptive algorithm that for any constant~$c$, asks $4dc\log d$ queries\footnote{all the complexities in this introduction are multiplied by $1+o(1)$ where the $o(1)$ is with respect to $d$.} and, with probability at least $1-\delta=1-1/d^{c-1}$, determines exactly the number of defective items. His algorithm, with the technique used in this paper\footnote{First estimate $d$ using the algorithm in this paper. Then determine $c$ to get success $\delta$ and run his algorithm}, gives a randomized Monte Carlo algorithm that asks $4d\log(d/\delta)$ queries with success probability at least $1-\delta$ for any $\delta$.

In this paper, we first give lower bounds for the number of queries. The first lower bound is $d\log(1/d\delta)$ for any $n$, $d$ and $\delta>1/n$. This shows that Cheng's algorithm is almost optimal (up to a multiplicative factor of $4$ and an additive term of $4d\log d$). For $\delta<1/n$, we give the tight bound $d\log (n/d)$, which, in particular, is the number of tests for any deterministic algorithm. This bound matches the tight bound for {\it finding} all the defective items.
We also give the better lower bound $d\log(1/\delta)$ for any large enough\footnote{$n\ge d^{\omega(1)}$ where $\omega(1)$ is with respect to $d$ for example $\log^*d$} $n$.

We then give a new randomized Monte Carlo algorithm that asks $d\log(d/\delta)$ queries. Our algorithm improves Cheng algorithm by a multiplicative factor of~$4$ and is optimal up to an additive term of $d\log d$. Notice that for $\delta=1/d^{\omega(1)}$ (especially when $\delta$ depends on $n$) our algorithm is optimal up to a small additive term.

Estimating the number of defective items is studied in~\cite{BshoutyBHHKS18,ChengX14,DamaschkeM10,DamaschkeM10b,FalahatgarJOPS16,RonT14}. The problem is to find an integer $D$ such that $d\le D\le (1+\epsilon)d$. In~\cite{BshoutyBHHKS18} Bshouty et al. modified Falhatgar et al. algorithm,~\cite{FalahatgarJOPS16}, and gave a randomized algorithm that makes expected number of $(1-\delta)\log\log d+O((1/\epsilon^2)\log(1/\delta))$ tests. They also prove the lower bound $(1-\delta)\log \log d+\Omega((1/\epsilon)\log(1/\delta))$.

\section{Definitions and Preliminary Results}
In this section we give some notations, definitions, the type of algorithms that are used in the literature and some preliminary results
\subsection{Notations and Definitions}
Let $X=[n]:=\{1,2,3,\ldots,n\}$ be a set of {\it items} with some {\it defective} items $I\subseteq [n]$. In Group testing, we {\it query} a subset $Q\subseteq X$ of items and the answer to the query is $Q(I):=1$ if $Q$ contains at least one defective item, i.e., $Q\cap I\not=\O$, and $Q(I):=0$, otherwise.

Let $I\subseteq [n]$ be the set of defective items. Let $\OO_I$ be an oracle that for a query $Q\subseteq [n]$ returns $Q(I)$. Let $A$ be an algorithm that has access to an {\it oracle} $\OO_I$. The output of the algorithm for an oracle $\OO_I$ is denoted by $A(\OO_I)$. When the algorithm is randomized, then we add the random seed $r$, and then the output of the algorithm is a random variable $A(\OO_I,r)$ in $[n]$. Let $A$ be a randomized algorithm and $r_0$ be a seed. We denote by $A(r_0)$ the deterministic algorithm that is equivalent to the algorithm $A$ with the seed $r_0$. We denote by ${\cal Q}(A,\OO_I)$ (or ${\cal Q}(A,\OO_I,r)$) the set of queries that $A$ asks with oracle $\OO_I$ (and a seed $r$). We say that the algorithm determines $|I|=d$ exactly if $A(\OO_I,r)=|I|$.

\subsection{Types of Algorithms}
In this paper we consider four types of algorithms that their running time is polynomial in $n$.
\begin{enumerate}
\item The {\it deterministic} algorithm $A$ with an oracle $\OO_I$, $I\subseteq X$. The query complexity of a deterministic algorithm $A$ is the {\it worst case complexity}, i.e, $\max_{|I|=d}|{\cal Q}(A,\OO_I)|$.
\item The {\it randomized Las Vegas algorithm}. We say that a randomized algorithm $A$ is a {\it randomized Las Vegas algorithm that has expected query complexity} $g(d)$ if for any $I\subseteq X$, $A$ with an oracle $\OO_I$ asks $g(|I|)$ expected number of queries and with probability $1$ outputs $|I|$.
\item The {\it randomized Monte Carlo algorithm}.
We say that a randomized algorithm $A$ is a {\it randomized Monte Carlo algorithm that has query complexity} $g(d,\delta)$ if for any $I\subseteq X$, $A$ with an oracle $\OO_I$ asks at most $g(|I|,\delta)$ queries and with probability at least $1-\delta$ outputs $|I|$.
\item The {\it randomized Monte Carlo algorithm with average case complexity}. We say that a randomized algorithm $A(r)$ is {\it Monte Carlo algorithm with average case complexity that has expected query complexity} $g(d,\delta)$ if for any $I\subseteq X$, $A$ asks $g(|I|,\delta)$ expected number of queries and with probability at least $1-\delta$ outputs $|I|$.
\end{enumerate}

\subsection{Preliminary Results}
We now prove a few results that will be used throughout the paper

Let $s\in \cup_{i=0}^\infty\{0,1\}^i$ be a {\it string} over $\{0,1\}$ (including the empty string $\lambda\in \{0,1\}^0$). We denote by $|s|$ the {\it length} of $s$, i.e., the integer $m$ such that $s\in \{0,1\}^m$. Let $s_1,s_2\in \cup_{n=0}^\infty\{0,1\}^n$ be two strings over $\{0,1\}$ of length $m_1$ and $m_2$, respectively. We say that $s_1$ is a {\it prefix} of $s_2$ if $m_1\le m_2$ and $s_{1,i}=s_{2,i}$ for all $i=1,\ldots,m_1$. We denote by $s_1\cdot s_2$ the {\it concatenation} of the two strings.

The following lemma is proved in~\cite{Bshouty18}.
\begin{lemma}\label{string} Let $S=\{s_1,\ldots ,s_N\}$ be a set of $N$ distinct strings such that no string is a prefix of another. Then
$$\max_{s\in S}|s|\ge E(S):=\E_{s\in S}[|s|]\ge \log N.$$
\end{lemma}

\begin{lemma}\label{Trivial} Let $A$ be a deterministic adaptive algorithm that asks queries. If
$A(\OO_I)\not= A(\OO_J)$ then there is $Q_0\in {\cal Q}(A,\OO_I)\cap {\cal Q}(A,\OO_J)$ such that $Q_0(I)\not=Q_0(J)$.

If, in addition, $I\subseteq J$ then $Q_0(I)=0$ and $Q_0(J)=1$.
\end{lemma}
\begin{proof} Since algorithm $A$ is deterministic, in the execution of $A$ with $\OO_I$ and $\OO_J$, $A$ asks the same queries as long as it gets the same answers to the queries. Since $A(\OO_I)\not= A(\OO_J)$, there must be a query $Q_0$ that is asked to both $\OO_I$ and $\OO_J$ for which $Q_0(I)\not=Q_0(J)$.\qed
\end{proof}

\begin{lemma}\label{mainlemma} Let $A$ be a deterministic adaptive algorithm that asks queries. Let $C\subseteq 2^{[n]}=\{I|I\subseteq [n]\}$.
If for every two distinct $I_1$ and $I_2$ in $C$ there is a query $Q\in {\cal Q}(A,\OO_{I_1})$ such that $Q(I_1)\not= Q(I_2)$ then
$$\max_{I\in C}|{\cal Q}(A,\OO_{I})|\ge \E_{I\in C}|{\cal Q}(A,\OO_{I})|\ge \log |C|.$$
That is, the average-case query complexity and the query complexity of $A$ is at least $\log |C|$.
\end{lemma}
\begin{proof} For $I\in C$ consider the sequence of the queries that $A$ with the oracle ${\cal O}_{I}$ asks and let $s(I)\in \cup_{n=0}^\infty\{0,1\}^n$ be the sequence of answers.
The query complexity and average-case complexity of $A$ is $s(C):=\max_{I\in C}|s(I)|$ and $\bar s(C):=\E_{I\in C}|s(I)|$ where $|s(I)|$ is the length of $s(I)$.
We now show that for every two distinct $I_1$ and $I_2$ in $C$, $s(I_1)\not=s(I_2)$ and $s(I_1)$ is not a prefix of $s(I_2)$.
This implies that $\{s(I)\ |\ I\in C\}$ contains $|C|$ distinct strings such that no string is a prefix of another. Then by Lemma~\ref{string}, the result follows.

Consider two distinct sets $I_1,I_2\subseteq [n]$. There is a query $Q_0\in {\cal Q}(A,\OO_{I_1})$ such that $Q_0(I_1)\not=Q_0(I_2)$. Consider the execution of algorithm $A$ with both ${\cal O}_{I_1}$ and ${\cal O}_{I_2}$, respectively. Since $A$ is deterministic, as long as the answers of the queries are the same both ($A$ with $\OO_{I_1}$ and $A$ with $\OO_{I_2}$) continue to ask the same query. Then, either we get to the query $Q_0$ in both execution and then $Q_0(I_1)\not=Q_0(I_2)$ or some other query $Q'$ that is asked before $Q_0$ satisfies $Q'(I_1)\not=Q'(I_2)$. In both cases $s(I_1)\not=s(I_2)$ and $s(I_1)$ is not a prefix of $s(I_2)$.\qed
\end{proof}

\section{Lower Bounds}
In this section, we prove some lower bounds for the number of queries that are needed to determine exactly the number of defective items with
a Monte Carlo algorithms

\begin{theorem}\label{TH1}
Let $\delta\ge 1/(2(n-d+1))$. Let $A$ be a randomized Monte Carlo adaptive algorithm that for any
set of defective items $I$ of size $|I|\in\{d,d+1\}$, with probability
at least $1-\delta$, determines exactly
the number of defective items~$|I|$. Algorithm $A$ must ask at least
$$d\log \frac{1}{2d\delta}-1$$ queries.

When $\delta\le 1/(2(n-d+1))$ then $A$ must ask at least $d\log(n/d)-O(d)$ queries which is the query complexity (up to additive term $O(d)$) of finding the defective items (and therefore, in particular, finding $|I|$) with $\delta=0$ error.
\end{theorem}
\begin{proof} Let $A(\OO_I,r)$ be a randomized Monte Carlo algorithm that for $|I|\in \{d,d+1\}$, determines $|I|$ with probability at least $1-\delta$
where $r$ is the random seed of the algorithm. Let $X(I,r)$ be a random variable that is equal to $1$ if $A(\OO_I,r)\not= |I|$ and $0$ otherwise. Then for any $I\subseteq [n]$,\ $\E_r[X(I,r)]\le \delta$. Let $m=\lfloor 1/(2\delta)\rfloor+d-1\le n$.
Consider any $J\subseteq [m]$, $|J|=d$. For any such $J$ let
$$Y_J(r)=X(J,r)+\sum_{i\in [m]\backslash J} X(J\cup\{i\},r).$$ Then for every $J\subseteq [m]$ of size $d$,
$\E_r\left[Y_J(r)\right] \le (m-d+1)\delta\le \frac{1}{2}.$ Therefore for a random uniform $J\subseteq [m]$ of size $d$ we have
$\E_r[\E_J[Y_J(r)]]=\E_J[\E_r[Y_J(r)]]\le 1/2$. Thus, there is $r_0$ such that
for at least half of the sets $J\subseteq [m]$, of size $d$, $Y_J(r_0)=0$. Let $C$ be the set of all $J\subseteq [m]$, of size $d$, such that $Y_J(r_0)=0$. Then $$|C|\ge \frac{1}{2}{m\choose d}=\frac{1}{2}{\lfloor 1/(2\delta)\rfloor+d-1 \choose d}.$$

Consider the deterministic algorithm $A(r_0)$.
We now claim that for every two distinct $J_1,J_2\in C$, there is a query $Q_0\in {\cal Q}(A(r_0),\OO_{J_1})$ such that $Q_0(J_1)\not=Q_0(J_2)$. If this is true then, by Lemma~\ref{mainlemma}, the query complexity of $A(r_0)$ is at least
$$\log |C|\ge \log \frac{1}{2}{\lfloor 1/(2\delta)\rfloor+d-1 \choose d}\ge d\log \frac{1}{2d\delta}-1.$$

We now prove the claim.
Consider two distinct $J_1,J_2\in C$. There is w.l.o.g $j\in J_2\backslash J_1$. Since $Y_{J_1}(r_0)=0$ we have $X(J_1,r_0)=0$ and $X(J_1\cup \{j\},r_0)=0$ and therefore $A(\OO_{J_1},r_0)=d$ and $A(\OO_{J_1\cup\{j\}},r_0)=d+1$. Thus, by Lemma~\ref{Trivial}, there is a query $Q_0\in {\cal Q}(A(r_0),\OO_{J_1})\cap {\cal Q}(A(r_0),\OO_{J_1\cup \{j\}})$ for which $Q_0(J_1)=0$ and $Q_0(J_1\cup\{j\})=1$. Therefore
$Q_0(\{j\})=1$ and then $Q_0(J_1)=0$ and $Q_0(J_2)=1$.\qed
\end{proof}

In the Appendix, we prove this lower bound for any randomized Monte Carlo algorithm with average-case complexity.

For large enough\footnote{The $\omega(1)$ is with respect to $d$. For example, $n>d^{\log^*d}$.} $n$, $n=d^{\omega(1)}$, the following result gives a better lower bound

\begin{theorem}\label{TH2}
Any randomized Monte Carlo adaptive algorithm that with probability
at least $1-\delta$ determines
the number of defectives must ask at least
$$\left(1-\frac{\log d+\log(1/\delta)+1}{\log n+\log(1/\delta)}\right)d\log \frac{1}{2\delta}$$ queries.

In particular, when $n=d^{\omega(1)}$ then the number of queries is at least
$$(1-o(1)) d\log\frac{1}{2\delta}$$
\end{theorem}
\begin{proof} Let $A(r)$ be a randomized Monte Carlo algorithm that determines the number of defective items with probability at least $1-\delta$
where $r$ is the random seed of the algorithm. Let $X'(I,r)$ be a random variable that is equal to $1$ if $A(\OO_I,r)\not= |I|$ and $0$ otherwise. Then for every $I$, $\E_r[X'(I,r)]\le \delta$. For every set $I$  and $i\in [n]\backslash I$ let $X(I,i,r)=X'(I,r)+X'(I\cup\{i\},r)$. Then for every $I\subseteq [n]$ and $i\in [n]\backslash I$,\ $\E_r[X(I,i,r)]\le 2\delta$. For $I$ of size $d$ chosen uniformly at random and $i\in [n]\backslash I$ chosen uniformly at random we have $\E_r\E_{I}\E_{i}[X(I,i,r)]=\E_{I}\E_{i}\E_r[X(I,i,r)]\le 2\delta$ and therefore there is exists a seed $r_0$ such that $\E_{I}\E_i[X(I,i,r_0)]\le 2\delta$. We now choose $q$ permutations $\phi_1,\ldots,\phi_q:[n]\to[n]$ uniformly and independently at random where
$$q=\left\lceil \frac{1+\log n}{\log \frac{1}{2\delta}}\right\rceil.$$
Then for any $I$ and $i\in [n]\backslash I$, $\phi_1(I),\ldots,\phi_q(I)$ are uniform and independent random sets of size $d$ and $\phi_1(i),\ldots,\phi_q(i)$ are uniform and independent random integers where $\phi_j(i)\not\in \phi_j(I)$ for all $j\in[q]$.
Hence
$$\E_{\{\phi_j\}_j}\left[\prod_{j=1}^q X(\phi_j(I),\phi_j(i),r_0)\right]=\prod_{j=1}^q \E_{\phi_j}\left[X(\phi_j(I),\phi_j(i),r_0)\right]\le (2\delta)^q$$ and
\begin{eqnarray*}
\E_{\{\phi_j\}_j}\E_{I}\E_{i}\left[\prod_{j=1}^q X(\phi_j(I),\phi_j(i),r_0)\right]&=&\E_{I}\E_{i}\E_{\{\phi_j\}_j}\left[\prod_{j=1}^q X(\phi_j(I),\phi_j(i),r_0)\right]\\&\le& (2\delta)^q.
\end{eqnarray*}
Therefore
$$\E_{\{\phi_j\}_j}\E_I\left[\sum_{i\in [n]\backslash I}\prod_{j=1}^q X(\phi_j(I),\phi_j(i),r_0)\right]\le (n-d)(2\delta)^q<\frac{1}{2}.$$
Thus there are permutations $\{\phi_j'\}_{j\in [q]}$ such that
$$\E_I\left[\sum_{i\in [n]\backslash I}\prod_{j=1}^q X(\phi_j'(I),\phi_j'(i),r_0)\right]< \frac{1}{2}.$$
Since $X$ takes values in $\{0,1\}$, this implies that for at least half of the sets $I\subseteq [n]$, of size $d$, and all $i\in [n]\backslash I$, there exists $j\in[q]$ such that
$
X(\phi_j'(I),\phi_j'(i),r_0)=0.
$
 Let $C$ be the class of such sets $I$ for which the later statement is true. Then
$$|C|\ge \frac{1}{2}{n\choose d}$$
and
\begin{eqnarray}\label{lll}
(\forall I\in C) (\forall i\in[n]\backslash I)(\exists j\in [q])\ X(\phi_j'(I),\phi_j'(i),r_0)=0.
\end{eqnarray}

Consider the following deterministic algorithm $A^*$: \\
$$ $$
\noindent\fbox{%
    \parbox{\textwidth}{%
{\bf Algorithm} $A^*$\\
For $j=1,\ldots,q$ \\
\hspace*{.5cm} Run $A(r_0)$ with oracle $\OO_I$  \\
\hspace*{.9cm}  If $A(r_0)$ asks $Q$ then ask the query
$\phi_j'^{-1} (Q)$.
}}
$$ $$

First notice that, if algorithm $A(r_0)$ has query complexity $M$, then $A^*$ has query complexity at most $qM$.

Now since $\phi_j'^{-1} (Q)\cap I=\O$ if and only if $Q\cap \phi_j'(I)=\O$, at iteration $j$ the algorithm $A(r_0)$ runs as if the defective items are $\phi_j'(I)$. Therefore, at iteration~$j$, the queries that are asked by $A(r_0)$ is ${\cal Q}(A(r_0),{\cal O}_{\phi_j'(I)})$ and the queries that are asked by $A^*$ is $\phi_j'^{-1}({\cal Q}(A(r_0),{\cal O}_{\phi_j'(I)}))$. Hence
\begin{eqnarray}\label{Eqv}
{\cal Q}(A^*,{\cal O}_I)=\bigcup_{j=1}^q \phi_j'^{-1}({\cal Q}(A(r_0),{\cal O}_{\phi_j'(I)})).
\end{eqnarray}

We now show that
\begin{claim}
For every two distinct sets $I_1,I_2\in C$ there is a query $Q'\in {\cal Q}(A^*,{\cal O}_{I_1})$ that gives different answers for $I_1$ and $I_2$.
\end{claim}
If this Claim is true then, by Lemma~\ref{mainlemma}, the query complexity $qM$ of $A^*$ is at least $\log|C|$ and then, since ${n\choose d}\ge (n/d)^d$,
\begin{eqnarray*}M&\ge& \frac{\log |C|}{q}\ge \frac{\log\frac{1}{2}{n\choose d}}{\frac{1+\log n}{\log(1/(2\delta))}+1}\\
&\ge & \frac{\log n-\log d-1}{\log n+\log\frac{1}{\delta}}\cdot d\log\frac{1}{2\delta}\\
&=&\left(1-\frac{\log d+\log(1/\delta)+1}{\log n+\log(1/\delta)}\right)d\log \frac{1}{2\delta}.
\end{eqnarray*}

We now prove the claim. Let $I_1,I_2\in C$ be two distinct sets of size $d$. Then there is $i_0\in I_2\backslash I_1$.
By (\ref{lll}) there is $j_0\in [q]$ such that for $\phi:=\phi_{j_0}'$,
 $X(\phi(I_1),\phi(i_0),r_0)=0$. Therefore $A(\OO_{\phi(I_1)},r_0)=|\phi(I_1)|=|I_1|$ and    $$A(\OO_{\phi(I_1)\cup \{\phi(i_0)\}},r_0)=|\phi(I_1)\cup\{\phi(i_0)\}|=|I_1|+1.$$
Therefore, by Lemma~\ref{Trivial}, there exists a query $Q_0\in {\cal Q}(A(r_0),\OO_{\phi(I_1)})$ that satisfies $Q_0(\phi(I_1))=0$ and $Q_0(\phi(I_1)\cup \{\phi(i_0)\})=1$.
That is, $Q_0\cap \phi(I_1)=\O$ and $Q_0\cap (\phi(I_1)\cup \{\phi(i_0)\})\not=\O$. This implies that $\phi(i_0)\in Q_0$. Since $\phi(i_0)\in \phi(I_2)$ we get that $Q_0\cap \phi(I_1)=\O$ and
 $Q_0\cap \phi(I_2)\not=\O$. Thus $\phi^{-1}(Q_0)\cap I_1=\O$ and
 $\phi^{-1}(Q_0)\cap I_2\not=\O$. That is, the query $Q':=\phi^{-1}(Q_0)$ satisfies
 $$Q'(I_1)\not=Q'(I_2).$$

 Now since $Q_0\in {\cal Q}(A(r_0),\OO_{\phi(I_1)})$, by (\ref{Eqv}), we have (recall that $\phi:=\phi_{j_0}'$) $$Q'=\phi^{-1}(Q_0)\in \phi^{-1}({\cal Q}(A(r_0),\OO_{\phi(I_1)}))\subseteq {\cal Q}(A^*,{\cal O}_I)$$ and therefore $Q'\in {\cal Q}(A^*,{\cal O}_I)$ . This completes the proof of the claim.\qed
\end{proof}

\section{Upper Bound}
In this section we prove
\begin{theorem}\label{TH5} There is a Monte Carlo adaptive algorithm that asks
$$d\log \frac{d}{\delta} +O\left(d+\log d\log\frac{1}{\delta}\right)=(1+o(1)) d\log \frac{d}{\delta}$$ queries and with probability at least $1-\delta$ finds the number of defective items.
\end{theorem}
This improves the bound $4d\log(d/\delta)$ achieved in~\cite{Cheng11}. By Theorem~\ref{TH2} this bound is optimal up to the additive term $(1+o(1))d\log d$.

We will use the following
\begin{lemma}\label{Find} \cite{ChengDX14,ChengDZ15,SchlaghoffT05} There is a deterministic algorithm, {\bf Find-Defectives},
that without knowing $d$, asks $d\log(n/d)+O(d)$ queries and finds the defective items.
\end{lemma}

Our algorithm, at the first stage, calls the procedure {\bf Estimate} that, with probability at least $1-\delta/2$ finds an estimate $D$ of the number defective items where $d\le D\le 8d$. This procedure makes $d+\log d\log(1/\delta)$ queries.

At the second stage, it uniformly at random partition the set of items $[n]$ into $t=D^2/\delta$ disjoint sets $B_1,B_2,\ldots,B_t$. We show that with probability at least $1-\delta/2$ each set contains at most one defective item. We call a set that contains a defective item a {\it defective set}. Therefore, with probability at least $1-\delta/2$, the number of defective items is equal to the number of defective sets. We then treat eash set $B_i$ as one item $i$ and call the algorithm {\bf Find-Defectives} in Lemma~\ref{Find} on $t$ items to find the defective sets. To do that, each test $Q\subseteq [t]$ in {\bf Find-Defectives} is simulated by the query $S(Q):=\cup_{i\in Q}B_i$ in our algorithm. Obviously, the set $Q$ contains an index of a defective set if and only if $S(Q)$ contains a defective item. Therefore, the algorithm will return the number of defective sets which, with probability at least $1-\delta/2$, is equal to the number of defective items. By Lemma~\ref{Find}, the number of queries asked in the second stage is
$$d\log \frac{t}{d}=d\log \frac{D^2/\delta}{d}=d\log \frac{d}{\delta}+O(d).$$

We now prove
\begin{lemma}\label{Eight} There is a Monte Carlo adaptive algorithm {\bf Estimate} that asks
$$O\left(d+\log d\log\frac{1}{\delta}\right)$$ queries and returns an integer $D$ that, with probability at least $1-\delta$, $D\ge d$ and, with probability $1$, $D \le 8d$.
\end{lemma}
\begin{proof} The algorithm is in Figure~\ref{Alg0}. For each $k=2^i$ the algorithm does the following $t=\lceil {2\log (1/\delta)}/{k}\rceil$ times: uniformly at random divides the items into $k$ mutually disjoint sets and then tests each set and counts (in the variable ``count'') the number of sets
that contains at least one defective item. First notice that for $k\le d$,
$$\Pr[\mbox{count}<k/4]\le {k\choose k/4}\left(\frac{1}{4}\right)^d\le {4^{k/4}}\left(\frac{1}{4}\right)^d\le2^{-d}.$$
The probability that the algorithm output $k< d$ is the probability that
the event ``$\mbox{count}<k/4$'' happens $t=\lceil {2\log (1/\delta)}/{k}\rceil$ times for some $k=2^i< d$. This is at most
$$\sum_{i=1}^{\lfloor\log d\rfloor}(2^{-d})^{\lceil {2\log (1/\delta)}/{2^i}\rceil}\le
\sum_{i=1}^{\lfloor\log d\rfloor} \delta^{2d/2^i}\le \delta.$$
Therefore with probability at least $1-\delta$ the output is greater or equal to $d$.
Since ``count'' is always less than or equal to the number of defective items $d$, when $4d< k\le 8d$, we have $count\le d<k/4$ and the algorithm halts.
Therefore the output cannot be greater than $8d$.

The number of queries that the algorithm asks is at most
$$\sum_{i=1}^{\lceil \log 8d\rceil} 2^i \left\lceil \frac{2\log (1/\delta)}{2^i}\right\rceil\le \sum_{i=1}^{\lceil \log 8d\rceil}{2\log (1/\delta)}+2^i=
O\left(d+\log d\log\frac{1}{\delta}\right).\qed$$
\end{proof}

\begin{figure}[h!]
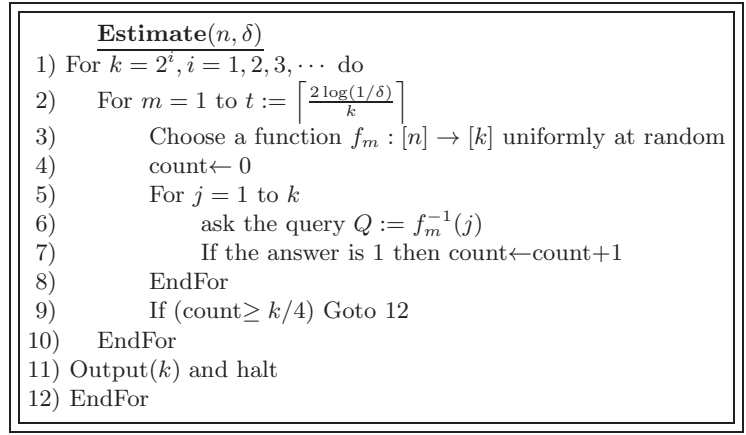

  \begin{center}
   \fbox{\fbox{\begin{minipage}{28em}
  \begin{tabbing}
  Xxxxx\=xxxx\=xxxx\=xxxx\= \kill
  \> \underline{{\bf Estimate$(n,\delta)$}}\\
  \ 1) For $k=2^i,i=1,2,3,\cdots$ do  \\
  \ 2) \> For $m=1$ to $t:=\left\lceil \frac{2\log (1/\delta)}{k}\right\rceil$\\
  \ 3) \>\> Choose a function $f_{m}:[n]\to [k]$ uniformly at random \\
  \ 4) \>\> count$\gets 0$\\
  \ 5) \>\> For $j=1$ to $k$ \\
  \ 6) \>\>\> ask the query $Q:=f_m^{-1}(j)$\\
  \ 7) \>\>\> If the answer is $1$ then count$\gets$count$+1$\\
 \  8) \>\> EndFor\\
  \ 9) \>\> If (count$\ge k/4$) Goto 12\\
   10) \> EndFor\\
   11) Output($k$) and halt\\
   12) EndFor
  \end{tabbing}\end{minipage}}}
  \end{center}
	\caption{An algorithm that estimates $d$.}
	\label{Alg0}
	\end{figure}

We now prove Theorem~\ref{TH5}.

\begin{proof}
The algorithm is in Figure~\ref{Alg1}. First, we run the algorithm in Figure~\ref{Alg0} that estimates $d$ and returns $D$ such that, with probability at least $1-\delta/2$, $d\le D\le 8d$.
Then the algorithm chooses a function $f:[n]\to [N]$ uniformly at random where $N=\lceil D^2/\delta\rceil$.
This is equivalent to uniformly at random divides the items into $N$ mutually disjoint sets $Q_i$, $i=1,\ldots,N$.
The probability that some $Q_i$ contains two defective items is
\begin{eqnarray*}
\Pr[(\exists i)\ Q_i\mbox{\ contains two defective items}]&=&1- \prod_{i=1}^{d-1}\left(1-\frac{i}{N}\right)\\
&\le & \sum_{i=1}^{d-1}\frac{i}{N}\le \frac{d^2}{2N}\le \frac{\delta}{2}.
\end{eqnarray*}
Then the algorithm run {\bf Find-Defectives} in Lemma~\ref{Find} on the $N$ disjoint sets $Q_1,\ldots,Q_N$ to find the number of sets that contains a defective items.
This number is with probability at least $1-\delta/2$ equal to the number of defective items.
Therefore with probability at least $1-\delta/2$ {\bf Find-Defectives} finds $d$.
This completes the proof of correctness.

The query complexity is the query complexity of Estimate$(n,\delta/2)$ and {\bf Find-Defectives} with $N$ items. This, by Lemma~\ref{Find} and~\ref{Eight} is equal to
$$d\log \frac{N}{d}+O(d) +O\left(d+\log d\log\frac{1}{\delta}\right)= d\log \frac{d}{\delta} +O\left(d+\log d\log\frac{1}{\delta}\right).\qed$$
\end{proof}

\begin{figure}[h!]
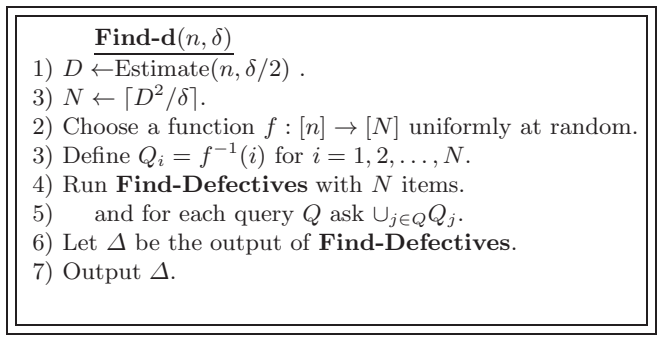

  \begin{center}
   \fbox{\fbox{\begin{minipage}{28em}
  \begin{tabbing}
  Xxxxx\=xxxx\=xxxx\=xxxx\= \kill
  \> \underline{{\bf Find-d$(n,\delta)$}}\\
  \ 1) $D\gets$Estimate$(n,\delta/2)$ . \\
  \ 3) $N\gets \lceil D^2/\delta\rceil$.\\
  \ 2) Choose a function $f:[n]\to [N]$ uniformly at random.\\
  \ 3) Define $Q_i=f^{-1}(i)$ for $i=1,2,\ldots,N$.\\
  \ 4) Run {\bf Find-Defectives} with $N$ items.\\
  \ 5) \> and for each query $Q$ ask $\cup_{j\in Q}Q_j$. \\
  \ 6) Let $\Delta$ be the output of {\bf Find-Defectives}.\\
  \ 7) Output $\Delta$.\\

  \end{tabbing}\end{minipage}}}
  \end{center}
	\caption{An algorithm that finds $d$.}
	\label{Alg1}
	\end{figure}

\bibliography{foo}
\bibliographystyle{plain}

\newpage
\section{Appendix}

\begin{theorem}\label{TH11}
Let $1/d^{\omega(1)}\ge \delta\ge 1/(2(n-d+1))$. Let $A$ be a randomized adaptive algorithm that for any
set of defective items $I$ of size $d$ or $d+1$, with probability
at least $1-\delta$, exactly determines
the number of defective items $|I|$. Algorithm $A$ must ask at least
$$(1-o(1))d\log \frac{1}{\delta}$$ expected number of queries.

When $\delta\le 1/(2(n-d+1))$ then $A$ must ask at least $(1-o(1))d\log n$ queries which is, asymptotically, the query complexity of finding the defective items with $\delta=0$ error.
\end{theorem}
\begin{proof} Let $A(r)$ be a randomized algorithm that for $I\subseteq [n]$, $|I|\in \{d,d+1\}$ and oracle $\OO_I$, determines $|I|$ with probability at least $1-\delta$
where $r$ is the random seed of the algorithm. Let $X(I,r)$ be a random variable that is equal to $1$ if $A(r,\OO_I)\not= |I|$ and $0$ otherwise. Then for any $I\subseteq [n]$,\ $\E_r[X(I,r)]\le \delta$. Let $m=\lfloor \tau/\delta\rfloor+d-1\le n$ where  $\tau=1/(d\log(1/(d\delta)))$.
Consider any $J\subseteq [m]$, $|J|=d$. For any such $J$ let
$$Y_J(r)=X(J,r)+\sum_{i\in [m]\backslash J} X(J\cup\{i\},r).$$ Then for every $J\subseteq [m]$ of size $d$,
$\E_r\left[Y_J(r)\right] \le (m-d+1)\delta\le \tau.$ Therefore for a random uniform $J\subseteq [m]$ of size $d$ we have
$\E_r[\E_J[Y_J(r)]]=\E_J[\E_r[Y_J(r)]]\le \tau$. Therefore, by Markov's inequality, for $\eta=1/\log(1/(d\delta))$,
$$\Pr_r[\E_J[Y_J(r)]>\eta]\le \frac{\tau}{\eta}=\frac{1}{d}.$$
That is, for random $r$, with probability at least $1-1/d$,
at least $1-\eta$ fraction of the sets $J\subseteq [m]$ of size $d$ satisfies $Y_J(r)=0$. Let $R$ be the set of seeds $r$ such that at least $1-\eta$ fraction of the sets $J\subseteq [m]$ of size $d$ satisfies $Y_J(r)=0$. Then $\Pr_r[R]\ge 1-1/d$. Let $r_0\in R$. Let $C_{r_0}$ be the set of all $J\subseteq [m]$ of size $d$ such that $Y_J(r_0)=0$. Then $$|C_{r_0}|\ge (1-\eta){m\choose d}=(1-\eta){\lfloor \tau/\delta\rfloor+d-1 \choose d}.$$

Consider the deterministic algorithm $A(r_0)$. As in Theorem~\ref{TH1},
for every two distinct $J_1,J_2\in C_{r_0}$, there is a query $Q\in {\cal Q}(A(r_0),\OO_{J_1})$ such that $Q(J_1)\not=Q(J_2)$. Then by Lemma~\ref{mainlemma}, the average-case query complexity of $A(r_0)$ is at least
$$\log |C_{r_0}|\ge \log (1-\eta){\lfloor \tau/\delta\rfloor+d-1 \choose d}\ge d\log \frac{\tau }{d\delta}-\log\frac{1}{1-\eta}.$$

Let $Z(\OO_I,r)=|{\cal Q}(A(r),\OO_I)|$. We have shown that for every $r\in R$, $$\E_{I\in C_{r}}[Z(\OO_I,r)]\ge d\log \frac{\tau}{d\delta}-\log\frac{1}{1-\eta}.$$
Therefore for every $r\in R$,
\begin{eqnarray*}
\E_{I}[Z(\OO_I,r)]&\ge& \E_{I}[Z(\OO_I,r)|I\in C_r]\Pr [I\in C_r]\\
&\ge& (1-\eta)\left( d\log \frac{\tau}{d\delta}-\log\frac{1}{1-\eta}\right).
\end{eqnarray*}
Therefore
\begin{eqnarray*}
\E_{I}\E_r[Z(\OO_I,r)]&=& \E_r\E_I[Z(\OO_I,r)]\\
&\ge &\E_r[\E_I[Z(O_I,r)|r\in R]\Pr[r\in R]]\\
&\ge& \left(1-\frac{1}{d}\right)(1-\eta) \left( d\log \frac{\tau}{d\delta}-\log\frac{1}{1-\eta}\right).
\end{eqnarray*}
Therefore there is $I$ such that
\begin{eqnarray*}
\E_r[Z(\OO_I,r)]\ge \left(1-\frac{1}{d}\right)(1-\eta) \left( d\log \frac{\tau}{d\delta}-\log\frac{1}{1-\eta}\right)
\end{eqnarray*}
and then
$$\E_r[Z(\OO_I,r)]\ge (1-o(1))d\log\frac{1}{\delta}.\qed$$
\end{proof}
\ignore{
New results - We can first estimate and then divide to blocks and find exactly the number by learning. In group testing one can give a simple learning - split to $\alpha d$ blocks - find in each block the defectives one by one. This gives $\alpha d+d\log(n/(\alpha d))$ - then find the best $\alpha$.
}

\end{document}